%% file: main.tex
\documentclass[letterpaper, 10 pt, conference]{ieeeconf}  % Comment this line out
\IEEEoverridecommandlockouts                              % This command is only
\usepackage{amsmath} % assumes amsmath package installed
\usepackage{amssymb}  % assumes amsmath package installed

\usepackage{amsthm}
\usepackage{bm}
\usepackage{mathtools}
\usepackage{xcolor}
\usepackage{accents}
\newlength{\dhatheight}

\usepackage{caption}
\usepackage{subcaption}

\usepackage{subfig}
\usepackage{booktabs}
\usepackage{adjustbox}
\usepackage[ruled,linesnumbered, noend]{algorithm2e}
\usepackage{bbm}
\usepackage{cite}
\usepackage{url}

\theoremstyle{definition}
\newtheorem{exmp}{Example}
\newtheorem{problem}{Problem}

\theoremstyle{definition}
\newtheorem{definition}{Definition}
\newtheorem{theorem}{Theorem}
\newtheorem{corollary}{Corollary}
\newtheorem{lemma}{Lemma}

\newtheorem{remark}{Remark}

\newcommand{\LL}[1]{\textcolor{blue}{[LL: #1]}}
\newcommand{\EF}[1]{\textcolor{red}{[EF: #1]}}
\newcommand{\AP}[1]{\textcolor{orange}{[AP: #1]}}

\title{\LARGE \bf
Uncertainty Propagation in Stochastic Systems via Mixture Models \\with Error Quantification
}

\author{Eduardo Figueiredo, Andrea Patane, Morteza Lahijanian, and  Luca Laurenti% <-this % stops a space
\thanks{E. Figueiredo and L. Laurenti are with the Delft Center for Systems and Control, TU Delft. Corresponding author's email:
        {\tt\small e.figueiredo@tudelft.nl}}%
\thanks{M. Lahijanian is with the Dept. of Aerospace Eng. Sciences and Computer Science, CU Boulder}%
\thanks{A. Patane is with the School of Computer Science and Statistics, Trinity College Dublin}%
}

\begin{document}

\maketitle
\thispagestyle{plain}
\pagestyle{plain}

%%%%%%%%%%%%%%%%%%%%%%%%%%%%%%%%%%%%%%%%%%%%%%%%%%%%%%%%%%%%%%%%%%%%%%%%%%%%%%%%
\begin{abstract}
Uncertainty propagation in non-linear dynamical systems has become a key problem in various fields including control theory and machine learning. In this work, we focus on discrete-time non-linear stochastic dynamical systems. We present a novel approach to approximate the distribution of the system over a given finite time horizon with a mixture of distributions. The key novelty of our approach is that it not only provides tractable approximations for the distribution of a non-linear stochastic system but also comes with formal guarantees of correctness. In particular, we consider the Total Variation (TV) distance to quantify the distance between two distributions and derive an upper bound on the TV between the distribution of the original system and the approximating mixture distribution derived from our framework. We show that in various cases of interest, including in the case of Gaussian noise, the resulting bound can be efficiently computed in closed form. This allows us to quantify the correctness of the approximation and to optimize the parameters of the resulting mixture distribution to minimize such distance.  The effectiveness of our approach is illustrated on several benchmarks from the control community.

\end{abstract}

%%%%%%%%%%%%%%%%%%%%%%%%%%%%%%%%%%%%%%%%%%%%%%%%%%%%%%%%%%%%%%%%%%%%%%%%%%%%%%%%
\section{Introduction}
Modern autonomous systems are becoming increasingly complex, commonly including non-linearities, data-driven components, and uncertain dynamics. The uncertainty affecting these systems can come from various sources including inherent probabilistic dynamics, lack of knowledge on parts of the model\cite{alamir2022learning}, and interactions with an uncertain environment \cite{pairet2021online}. 
Consequently, when autonomous systems need to be employed for safety-critical applications, where failure may translate into costly or deadly accidents, it becomes paramount to consider models that include non-linear stochastic dynamics. Unfortunately, for non-linear stochastic dynamical models, the distribution of the system at a given time cannot be computed in closed-form \cite{girard2002gaussian}. % \AP{I would give more emphasis to the sentence that follows, as it is the core motivation of the paper. Also, I feel to the best of our knowledge is almost mandatory for this sort of statements} %\AP{I also feel we might want to give more details before making such a definitive statement. Some sort of guarantees, say on the mean, the variance, or reachability probabilities and such, do already exist – and it is an easy comeback from reviewers. If we make it clear that we mean guarantees over the full distribution or something to that effect we defend ourselves from cheap criticism. You explain these things later on, so maybe just a slight reordering would be effective}
%and for these systems approximate methods for uncertainty propagation with error quantification are missing \cite{}. 
This leads to the main research question of this paper: Can we develop an efficient framework for uncertainty propagation in non-linear stochastic dynamical systems with formal guarantees of correctness?

Various methods have been proposed to propagate uncertainty in non-linear stochastic systems \cite{landgraf2023probabilistic}. These either consider analytical approximation methods, such as Taylor expansions or Stirling's interpolation \cite{einicke2012smoothing, schei1997finite}, to approximate the one-step dynamics of the systems, or consider numerical approximation of the various integrals \cite{dunik2020state, ito2000gaussian} and/or only propagate discrete samples from the system distribution \cite{lui1998sequential, arulampalam2002tutorial}. However, none of these methods comes with error bounds on the approximation error. An exception is \cite{polymenakos2020safety}, which is however only limited to Gaussian processes with additive noise and can only guarantee that with high probabilities the system trajectories will stay close to the mean. An alternative approach proposed in \cite{jasour2021moment} is to rely on moment matching to propagate the moments of non-linear stochastic dynamical systems with non-linearities that are either polynomial or trigonometric functions. 
However, in practice, it is often required to propagate the full distribution, and not just the moments, e.g., for chance constraint satisfaction in planning problems \cite{pairet2021online}. Additionally, we should stress that, because of their universal approximation properties, various papers already attempted to use mixture distributions, and Gaussian Mixture Models (GMMs) in particular, to approximate the distribution of a stochastic dynamical system \cite{weissel2009stochastic, terejanu2008uncertainty}. However, these works are mostly based on heuristics on how to build the mixture and come with no guarantees of correctness. 

%In general, those methods either opt to approximate the dynamics $g$ by e.g. Taylor expansions \cite{} or Stirling's interpolation \cite{}, or to approximate the distribution $p_{x_t}$ by moment matching \cite{jasour2021moment} \cite{}, GMMs \cite{}, or point mass filters \cite{}. Although some of those methods have been successfully employed for planning tasks, they lack, to the best of our knowledge, formal guarantees on the distance of the approximation to the actual probability measure being approximated.

In this work, given an error threshold $\delta >0$ and a time horizon $T,$ we develop a framework that for each time step $t\leq T$  builds a mixture distribution approximation of a given non-linear stochastic dynamical system, which is guaranteed to be $\delta-$close to the distribution of the system at that time. To quantify the closeness between two distributions we consider the Total Variation (TV) distance, defined as the largest absolute difference between the probabilities that two probability distributions assign to an event \cite{gibbs2002choosing}. Interestingly, we show that by appropriately picking the components of the approximating mixture distribution, an upper bound for its TV from the unknown distribution of the system at time $t$ can be computed in closed form. Using this bound, we then propose a refinement algorithm that iteratively builds mixture models by efficiently placing distributions in regions of the state space by minimizing the TV bound. We empirically illustrate the efficacy of our framework in various benchmarks.

In summary, the main contributions of this work are: i) a novel framework for computing a mixture distribution approximating the distribution of a stochastic dynamical system at a given time, ii) formal bounds of correctness in terms of the TV distance and a proof of convergence of the approximating error to $0$,  iii) an adaptive grid refinement algorithm for bound optimization, and iv) empirical evaluation of the efficacy of our framework on four case studies including a Dubin's car model and an application on a planning problem.

%%%%%%%%%%%%%%%%%%%%%%%%%%%%%%%%%%%%%%%%%%%%%%%%%%%%%%%%%%%%%%%%%%%%%%%%%%%%%%%%

%%%%%%%%%%%%%%%%%%%%%%%%%%%%%%%%%%%%%%%%%%%%%%%%%%%%%%%%%%%%%%%%%%%%%%%%%%%%%%%%
\section{Problem Formulation}

\subsection{Notation}

Given a Borel measurable space $\mathcal{X} \subseteq \mathbb{R}^q$ we denote by $\mathcal{P}(\mathcal{X})$ the set of probability distributions on $\mathcal{X}$.
For a random variable $x_{t}$ taking values in $\mathcal{X}$,  we denote by $\mathbb{P}_{x_t} \in \mathcal{P}(\mathcal{X})$ the probability measure associated to $x_t$, 
and by $p_{x_t}$ the probability density function associated to $\mathbb{P}_{x_t}$. For all the probability distributions considered in this paper, we always assume that they have a density. 
Additionally, given a set of $K$ probability distributions $\mathbb{P}_{\theta_1},...,\mathbb{P}_{\theta_K}\in \mathcal{P}(\mathcal{X})$ and a set of weights $\omega_1,...,\omega_K \in [0,1]$ such that $\sum_{k=1}^K \omega_k = 1$, we define a \emph{mixture distribution} of size $K>0$ as the distribution $\hat{\mathbb{P}} =  \sum_{k=1}^K  \omega_k \mathbb{P}_{\theta_k}$. If each $\mathbb{P}_{\theta_k}$ is a Gaussian distribution, then $\hat{\mathbb{P}}$ is called a \emph{Gaussian mixture model} (GMM) or \emph{Gaussian mixture distribution}.
Finally, for a probability distribution ${\mathbb{P}} \in \mathcal{P}(\mathcal{X})$ and a measurable function $g: \mathcal{X} \rightarrow \mathcal{Y}$, where $\mathcal{Y} \subseteq \mathbb{R}^{d}$, we denote the push-forward measure of ${\mathbb{P}}$ by $g$ as $g \# {\mathbb{P}}$ such that for all $A \subset \mathcal{Y}$, $(g \# \mathbb{P})(A) := \mathbb{P}(g^{-1}(A))$. We note that $g \# \mathbb{P}$ is still a probability distribution such that $g \# \mathbb{P} \in \mathcal{P}(\mathcal{Y})$.

\subsection{System Description}

%\EF{Changed the function to g so we can keep the dynamics on x denoted as f (since all the proofs are using f, I thought this would be better)}

We consider a  discrete-time stochastic process described by the following stochastic difference equation
\begin{align} \label{eq:original-problem}
    {x}_{t+1} = g({x}_{t}, {\varepsilon}_{t}) \text{,}\qquad  {x}_{0} \sim \mathbb{P}_{x_0}, \, {\varepsilon}_{t} \sim \mathbb{P}_{ {\varepsilon}}, 
\end{align}
where $g: \mathbb{R}^{d} \times \mathbb{R}^{q} \rightarrow \mathbb{R}^{d}$ is a possibly non-linear continuous function representing the one-step dynamics of System \eqref{eq:original-problem}, $x_0$ is a random initial condition with distribution $\mathbb{P}_{x_0}\in \mathcal{P}(\mathbb{R}^d)$, and $ {\varepsilon}_t$ is a noise term distributed according to a absolutely continuous distribution independent and identically distributed at each time step\footnote{We remark that the assumption of time invariance of the noise distribution is only introduced to simplify notation and our methods naturally extend to time-varying noise distributions, as we will emphasize in Remark \ref{rmk:homoscedasticity-variance} in Section \ref{section:approx-scheme}. } $\mathbb{P}_{ {\varepsilon}}\in \mathcal{P}(\mathbb{R}^q)$. Intuitively, System \eqref{eq:original-problem} represents a general model of a $d-$dimensional discrete-time stochastic system with possible non-linearities both in the state and noise.

In order to study the time evolution of  System \eqref{eq:original-problem}, we first introduce $\mathbb{T}_{x}$, the one-step transition kernel of System \eqref{eq:original-problem}, which for any time $t$ represents the probability distribution of $x_t$ given that $x_{t-1}=x$ and is formally defined as:
$$\mathbb{T}_{x}={g}( {x},  \cdot ) \# \mathbb{P}_{ {\varepsilon}}.$$
For simplicity, with an abuse of notation, we denote by $p(y \mid x)$ the probability density associated to $\mathbb{T}_{x}$.

\begin{exmp} \label{ex:gaussian-noise-explanation}
In the case where $g(x,\varepsilon)=f(x)+ \varepsilon$ and the noise distribution is $\mathbb{P}_{ {\varepsilon}} = \mathcal{N}( {0}, \Sigma_{ {\varepsilon}})$, that is, a zero mean Gaussian with covariance matrix $\Sigma_{ {\varepsilon}}$,  we have that $\mathbb{T}_{x} = \mathcal{N}(f( {x}), \Sigma_{ {\varepsilon}})$, a Gaussian distribution with mean $f(x)$ and covariance matrix $\Sigma_{ {\varepsilon}}$.  
\end{exmp}

The probability distribution of System \eqref{eq:original-problem} at time  $t$, $\mathbb{P}_{x_t}$, can be described by its density $p_{{x}_{t}},$ which can be computed iteratively over time by using the following Chapman-Kolmogorov equation \cite{girard2002gaussian,weissel2009stochastic}:
\begin{align}
& \label{Eqn:ChapmanKolmogorov}
p_{x_t}(y) = \int_{\mathbb{R}^d} p(y\mid x) p_{x_{t-1}}(x) dx.
\end{align}
That is, $p_{{x}_{t}}$ is obtained by marginalizing the one-step transition kernel w.r.t. $p_{{x}_{t-1}}.$
Unfortunately, even though  $p(y\mid x)$ can be expressed in closed form in many cases of interest,  Eqn 
\eqref{Eqn:ChapmanKolmogorov} is in general not tractable. For instance, this is the case for the system in Example \ref{ex:gaussian-noise-explanation} where $p(y\mid x)$ is Gaussian, but $p_{x_t}$ is intractable and non-Gaussian if $f$ is non-linear. Consequently, approximation schemes are required.

\subsection{Problem Statement}
Our goal in this paper is to design an approximating mixture distribution $\hat{\mathbb{P}}_{x_t}$, which is close to the actual probability distribution $\mathbb{P}_{x_t}$, but that is also tractable and can be computed analytically in closed form. %\AP{I don't understand what it means for an approximating distribution to admit closed-form solutions – maybe you mean closed form? Or that can be propagated in closed-form or something?}. 
To do that, we need to introduce a metric to quantify the distance between probability distributions. While many metrics have been introduced for this task \cite{gibbs2002choosing}, in this paper, we focus on the Total Variation distance.

\begin{definition}[Total Variation] \label{def:tv-distance}
Let  $(\mathcal{X}, \mathcal{B}(\mathcal{X}))$ be a measurable space, and  $\mathbb{P},\mathbb{Q} \in \mathcal{P}(\mathcal{X})$ be probability distributions. Then, the Total Variation distance between  $\mathbb{P}$ and $\mathbb{Q}$ is defined as:
$$
\text{TV}(\mathbb{P}, \mathbb{Q}) := \sup_{A \in \mathcal{B}(\mathcal{X})} |\mathbb{P}(A) - \mathbb{Q}(A)|.
$$
%When probability densities exist, it is a known result  that this definition is equivalent to:
%See Lemma 2.1 (page 84)
%
%\begin{equation} \label{def:tv}
%\text{TV}(\mathbb{P}, \mathbb{Q}) = \frac{1}{2} \int_{\mathcal{X}} |p(x) - q(x)|dx
%\end{equation}
\end{definition}

Intuitively, the TV is the largest difference that the probability of the same event can have according to two probability distributions. Consequently, because of its intuitive and worst-case nature, the TV represents a particularly convenient choice to quantify the distance between the distribution of stochastic dynamical systems. % In fact, these systems are often employed in applications where probability constraints must be satisfied, e.g. MPC \cite{landgraf2023probabilistic}, \cite{} \EF{I guess you were also thinking about the ETH paper here?},  and 
We should also stress that, differently from the TV case,  closeness in other measures commonly employed, such as Wasserstein distance, does not necessarily imply closeness in the probability of events \cite{gibbs2002choosing}.

We are finally ready to formulate
our problem statement.

\begin{problem}
\label{prob:MainProb}
Consider a finite time horizon  $ \{ 1, ..., T \} \subset \mathbb{N}$, and a threshold $\delta > 0.$ Find a set of mixture distributions of finite size $\hat{\mathbb{P}}_{x_1},...,\hat{\mathbb{P}}_{x_T}$ such that for all $t\in \{ 1, ..., T \}$ it holds that:
\begin{equation} \label{eq:tv-threshold}
    \text{TV} \left( \mathbb{P}_{x_t}, \hat{\mathbb{P}}_{x_t} \right) \leq \delta.
\end{equation}
\end{problem}

Note that Problem \ref{prob:MainProb} is stated for a general mixture distribution $\hat{\mathbb{P}}_{x_t}$. This is because depending on the one-step transition kernel of System \eqref{eq:original-problem}, one may want to focus on different classes of mixture distributions. Nevertheless,  because of their modeling flexibility and because of the existence of TV bounds between two Gaussian distributions \cite{devroye2018total}, in this paper we will give particular attention to GMMs.

We should also note that solving Problem \ref{prob:MainProb} is particularly challenging not only because Eqn \eqref{Eqn:ChapmanKolmogorov} cannot be solved in closed form, but also because it requires propagating over time the uncertainty introduced by each approximation step. In fact, according to Eqn \eqref{Eqn:ChapmanKolmogorov}, to compute $\mathbb{P}_{x_t}$, one needs to marginalize the one step transition kernel wrt $\mathbb{P}_{x_{t-1}},$ which is unknown for any $t>0$.

\paragraph{Approach}
To solve Problem \ref{prob:MainProb}, we propose a stochastic approximation scheme consisting of an iterative process based on the approximation of Eqn \eqref{Eqn:ChapmanKolmogorov} with a mixture distribution. Intuitively, we partition $\mathbb{R}^d$ in various subspaces $\mathcal{X}_k$ and select representative points $x^{(k)} \in \mathcal{X}_k$ such that $\int_{\mathbb{R}^d} p(y \mid x) p_{x_t}(x)dx$ is close to $\sum_{k=1}^{K} p(y \mid x^{(k)}) \int_{\mathcal{X}_k} p_{x_t}(x)dx$. %, which are basically mixtures of the one-step transition kernel $\mathbb{T}_{x}$ at various representative points. Similar approximation schemes can be found in \cite{landgraf2023probabilistic} (Section 2.9) or \cite{terejanu2008uncertainty}. The novelty in our work is the ability to prove probabilistic guarantees on the approximation.
 In Section \ref{section:approx-scheme}, we detail our stochastic approximation scheme, and we state proof of its correctness and convergence to the distribution of System \eqref{eq:original-problem}. In Section \ref{section:algorithm}, we present the resulting algorithm used to build the mixture approximations. Finally, in Section \ref{section:exp-results}, we present numerical experiments.

\section{Stochastic Approximation Scheme} \label{section:approx-scheme}

To solve Problem \ref{prob:MainProb}, we introduce a stochastic approximation scheme based on iteratively approximating the distribution of System \eqref{eq:original-problem} by mixtures of the one-step transition kernel. 
In particular, for each time step $t$ we consider a partition of $\mathbb{R}^d$ in $K_t$ non-overlapping regions $\bm{\mathcal{X}}_{t} := \{ \mathcal{X}^{(1)}_{t}, ..., \mathcal{X}^{(K_t)}_{t} \}$ and for each $\mathcal{X}^{(k)}_{t} \in \bm{\mathcal{X}}_{t}$ we consider a representative point $x^{(k)}_{t} \in \mathcal{X}^{(k)}_{t}.$ Then, we define the approximating mixture distribution $\hat{\mathbb{P}}_{x_t}$ iteratively as follows:
\begin{align} \nonumber
    &  \hat{\mathbb{P}}_{x_0} := \mathbb{P}_{x_0} \\
    &  \hat{\mathbb{P}}_{x_t} := \sum_{k = 1}^{K_{t-1}} \omega^{(k)}_{t} {\mathbb{T}}_{x^{(k)}_{t-1}} \label{eq:approximation-scheme}
\end{align}
with $\omega^{(k)}_{t} := \hat{\mathbb{P}}_{x_{t-1}}({\mathcal{X}^{(k)}_{t-1}}).$ Intuitively, at each time $t$, Eqn \eqref{eq:approximation-scheme} builds a mixture approximation of $\mathbb{P}_{x_t}$ by averaging the one-step transition kernel of System \eqref{eq:original-problem} with weights that depend on the approximating mixture distribution at the previous time step. 

\begin{exmp} \label{ex:propagation-example-gaussian}
For the system dynamics of Example \ref{ex:gaussian-noise-explanation}, Eqn \eqref{eq:approximation-scheme} leads to a mixture distribution approximation at time step $t$ given by the GMM $\hat{\mathbb{P}}_{x_t} = \sum_{k=1}^{K_{t-1}} \omega^{(k)}_{t} \mathcal{N}(f(x^{(k)}_{t-1}), \Sigma_{\varepsilon})$, where $\omega^{(k)}_{t} = \hat{\mathbb{P}}_{x_{t-1}}({\mathcal{X}^{(k)}_{t-1}})$ is the probability mass of the GMM $\hat{\mathbb{P}}_{x_{t-1}} = \sum_{k=1}^{K_{t-2}} \omega^{(k)}_{t-1} \mathcal{N}(f(x^{(k)}_{t-2}), \Sigma_{\varepsilon})$ (built in the previous step) within the region $\mathcal{X}^{(k)}_{t-1}$.
\end{exmp}

\begin{remark}
Note that the mixture $\hat{\mathbb{P}}_{x_t}$ as defined in Eqn \eqref{eq:approximation-scheme} is not necessarily a Gaussian mixture. If for the particular application it would be preferable to obtain a GMM, then this can always be obtained by approximating each ${\mathbb{T}}_{x^{(k)}_{t-1}}$ in Eqn \eqref{eq:approximation-scheme} with a Gaussian (mixture) distribution. As the total variation is a distance that satisfies the triangle inequality, then the bounds on the approximation error developed in the rest of this section would hold also for this setting, with a correction term depending on the closeness of each  ${\mathbb{T}}_{x^{(k)}_{t-1}}$ to the respective Gaussian or Gaussian mixture approximation.
\end{remark}

It is obvious that the quality of the approximation obtained from Eqn \eqref{eq:approximation-scheme} depends on the choice of the weights and representative points. How to perform this choice is the object of Section \ref{section:algorithm}. In the remainder of this section, we focus on quantifying $\text{TV}(\hat{\mathbb{P}}_{x_t},{\mathbb{P}}_{x_t})$. 

%Intuitively, as the number of partitions grow (and, consequently, the regions $\mathcal{X}^{(k)}_{t}$ become smaller), $\hat{p}_{x_t}$ should get closer to $p_{x_t}$ by a similar reasoning applied to the construction of Lebesgue integrals using simple functions. In Section XX (\EF{SEE SECTION AND ADD}), we prove that, under the assumptions stated in the Problem Formulation, $\hat{p}_{x_t}$ converges to $p_{x_t}$ as partitions become finer. 

\subsection{Total variation bounds}

We start our analysis with Theorem \ref{thm:general-case}, where we bound $\text{TV}(\hat{\mathbb{P}}_{x_t},{\mathbb{P}}_{x_t})$. Proofs for Theorems and Lemmas of this Section can be found in the Section \ref{section:proofs}.

\begin{theorem} \label{thm:general-case}
Let function $s \big( x, x^{(k)} \big)$ be defined as
\begin{equation}
s \big( x, x^{(k)} \big) := \frac{1}{2} \int_{\mathbb{R}^d} \big| p(y \mid x) - p(y \mid x^{(k)}) \big| dy
\end{equation}
Then, for any $t>0$ it holds that 
\begin{equation} \label{eq:tvBoundGeneralCase}
\begin{split}
& \text{TV}  \left( \mathbb{P}_{x_t}, \hat{\mathbb{P}}_{x_t} \right) \leq \text{TV} \left( \mathbb{P}_{x_{t-1}}, \hat{\mathbb{P}}_{x_{t-1}} \right) + \\ 
&  \qquad \sum_{k=1}^{K_{t-1}}  \max_{{x}\in \mathcal{X}^{(k)}_{t-1}} s\big( {x}, x^{(k)}_{t-1} \big)  \hat{\mathbb{P}}_{x_{t-1}}(\mathcal{X}^{(k)}_{t-1}),
\end{split}
\end{equation}
where $\text{TV} \left( \mathbb{P}_{x_0}, \hat{\mathbb{P}}_{x_0} \right)=0$.
\small
\normalsize
\end{theorem}

Theorem \ref{thm:gaussian-case} guarantees that the approximation error at each time step propagates at most linearly over time.    Furthermore, because of term $\hat{\mathbb{P}}_{x_{t-1}}(\mathcal{X}^{(k)}_{t-1})$, only regions with non-negligible probability mass w.r.t. $\hat{\mathbb{P}}_{x_{t-1}}$ influence the TV bound.  The key component of Theorem 1 is the function $s\big( {x}, x^{(k)}_{t-1} \big)$, which is equivalent to $\text{TV}(\mathbb{T}_{x},\mathbb{T}_{x^{(k)}_{t-1}})$, i.e., the total variation between the one-step transition kernel at $x$ and $x^{(k)}_{t-1}$.
Consequently, the computation of $s\big( {x}, x^{(k)}_{t-1} \big)$ depends on the specific noise distribution in System \eqref{eq:original-problem}. 

Luckily, closed-form expressions for $s\big( {x}, x^{(k)}_{t-1} \big)$ or for upper bounds exist for many cases of interest, including when $\mathbb{T}_{x}$ follows a Gaussian, Gamma, or uniform distribution. 
Because of its widespread use, the Gaussian case is explicitly considered in detail in the next paragraph. For the Gamma distribution, one can use the fact that $s$ can be bounded by the KL divergence using Pinsker's inequality and then rely on closed forms for the KL divergence between Gamma distributions \cite{tsybakov2009nonparametric}. For the uniform distribution, computing $s$ can be performed by separating the integral in regions based on the overlapping of the supports (which can be easily done when the dimensions are independent, and hence the supports are hyper-rectangles). 

\paragraph{Total Variation Bounds for the Gaussian Case}

The following result is a corollary of Theorem \ref{thm:general-case} in the case of additive Gaussian noise.
\begin{corollary}\label{thm:gaussian-case}
Assume that $g(x_t,\bm{\varepsilon}_{t}) = f(x_t) + \bm{\varepsilon}_{t}$ and $\bm{\varepsilon}_{t} \sim \mathcal{N}(0, \Sigma_{\varepsilon})$. Then, for any $t>0$ it holds that 
\begin{equation} \label{eq:tvBoundGaussianCase}
\begin{split}
& \text{TV}  \left( \mathbb{P}_{x_t}, \hat{\mathbb{P}}_{x_t} \right) \leq \text{TV} \left( \mathbb{P}_{x_{t-1}}, \hat{\mathbb{P}}_{x_{t-1}} \right) + \\ 
&  \qquad \sum_{k=1}^{K_{t-1}} \text{erf} \left( \max_{{x}\in \mathcal{X}^{(k)}_{t-1}} h \big( {x}, x^{(k)}_{t-1} \big) \right)  \hat{\mathbb{P}}_{x_{t-1}}(\mathcal{X}^{(k)}_{t-1}),
\end{split}
\end{equation}
where $\text{TV} \left( \mathbb{P}_{x_0}, \hat{\mathbb{P}}_{x_0} \right)=0$ and
\small
$$
h \big( x, x^{(k)}_{t-1} \big) := \frac{1}{2 \sqrt{2}} \sqrt{(f \left( x \right) - f(x^{(k)}_{t-1}))^T \Sigma_{\varepsilon}^{-1} (f(x) - f(x^{(k)}_{t-1}))}.
$$
\normalsize
%\EF{As the theorem needs the inverse of the covariance, I think we are safer if we state in this "less" general setting. I changed the mean back to 0 because it seems better to me and we can always sum the bias in f, but if you think it is better to present with mu, no issue for me :)}
%\ml{this looks good to me.}
\end{corollary}
Eqn \eqref{eq:tvBoundGaussianCase} follows from Theorem \ref{thm:general-case} using the fact that when the one-step transition kernel is Gaussian, then $s\big( x, x^{(k)}_{t-1} \big)$ can be computed in a closed form in terms of the function $h \big( x, x^{(k)}_{t-1} \big)$. Note also that for any $x\in \mathcal{X}_{t-1}^{(k)}$ the function $h \big( x, x^{(k)}_{t-1} \big)$, and consequently the tightness of the bounds, depends on the following factors: (i)  $f(x) - f(x^{(k)}_{t-1})$, which quantifies how much $f$ varies locally,  (ii) $\Sigma_{\varepsilon}^{-1}$, which suggests that lower noise variances may lead to worse bounds. This should not appear surprising when considering the definition of TV in Definition \ref{def:tv-distance}. In fact, from the definition of TV it follows that the TV between two delta Dirac distributions is one unless they have exactly the same mean. Consequently, the smaller the noise variance, the more local variation of $f(x)$ from the mean of each component of the GMM will affect the bound.

\begin{remark} \label{rmk:homoscedasticity-variance}
Corollary \ref{thm:gaussian-case} can be straightforwardly generalized to the case where the covariance matrix varies with time, as long as it is independent of $x$. If heteroscedasticity is required in the noise model, one possible strategy is to adapt the proof using the fact that KL divergence can be used to upper bound the TV distance and use Prop. 2.1 of \cite{devroye2018total} to bound term $s$.
\end{remark}

For both Theorem \ref{thm:general-case} and Corollary \ref{thm:gaussian-case}, at each time step the TV bound is obtained by a sum over all components in the mixture. An interesting question is whether the bound decreases with a mixture of a larger size. We answer positively to this question in the following subsection, where we prove the convergence of our method: the $\text{TV}  \left( \mathbb{P}_{x_t}, \hat{\mathbb{P}}_{x_t} \right)$ goes to zero for approximating mixture models of arbitrarily large size.

\subsection{Convergence Analysis}

In Theorem \ref{thm:convergence} 
below we show that the approximation error introduced by Eqn \eqref{eq:approximation-scheme} can be made arbitrarily small by increasing the size of the mixture. 
\begin{theorem}\label{thm:convergence} 
For $\bar{x}\in \mathbb{R}^d$ assume that $\mathbb{T}_{\bar{x}}$-almost surely (a.s.) for every $y \in \mathbb{R}^d$ it holds that
$$
\|x - \bar{x}\|_{\infty} \rightarrow 0 \implies p(y \mid x) \rightarrow p(y \mid \bar{x}),
$$ 
where $\|x - \bar{x}\|_{\infty}$ is the infinity norm of vector $x - \bar{x}.$
Then,  for any $\epsilon > 0$ and any time $t<\infty$, there exists a sequence of partitions $\{ \bm{\mathcal{X}}_{0}, ..., \bm{\mathcal{X}}_{t-1} \}$ such that 
\begin{equation}
\text{TV} \left( \mathbb{P}_{t}, \hat{\mathbb{P}}_{t} \right) \leq \epsilon.
\end{equation}

\end{theorem}
Note that the assumption of almost surely continuity of $p(y \mid x)$ in Theorem \ref{thm:convergence} is not particularly restrictive. For instance, it holds for distributions with continuous density, e.g., Gaussian, or even with bounded support, where the boundary of the support has measure zero. %The assumption rules out discrete distributions though.
%For bounded support distributions $\mathbb{T}_{x}$ (e.g. continuous uniform), we note that there are discontinuities at the boundaries of $|p(y \mid x) - p(y \mid \bar{x})|$ even when $g$ is continuous. As long as this boundaries have zero Lebesgue measure, the theorem guarantees the convergence.

\begin{remark}
We should remark that the convergence in Theorem \ref{thm:convergence} is proved assuming that each $\bm{\mathcal{X}}_{i}$ is a uniform partition of a compact set containing $1-\frac{\epsilon}{2 t}$ probability mass of $\hat{\mathbb{P}}_{x_{i}}.$ While sufficient for showing convergence, this choice of partition is obviously sub-optimal, as in general, we would like to place more distributions from the mixture in regions where there is more probability mass. This is explored in our algorithmic framework detailed in the next section.
\end{remark}

\section{Algorithm} \label{section:algorithm}

We summarize our overall procedure to solve Problem \ref{prob:MainProb} in Algorithm \ref{alg:main-algo}. For every time $t,$ in Lines 2-5 we compute the approximating mixture at the current time step following Eqn \eqref{eq:approximation-scheme}. 
Then, in Line 6 we identify a hyper-rectangle $\mathcal{S}_t$ containing at least $1-\epsilon$ probability mass of  $\hat{\mathbb{P}}_{x_{t}}$, where $\epsilon$ is a given parameter, which in our implementation, for time step $t$ is taken as $\epsilon \ll \frac{\delta}{T}$, such that the impact of not partitioning region $\mathbb{R}^d\setminus \mathcal{S}_t$ on the bound is small compared to our desired precision. Such a region can either be identified analytically as $\hat{\mathbb{P}}_{x_{t}}$ is known, e.g., in the case $\hat{\mathbb{P}}_{x_{t}}$ is a GMM, or using samples. 
$\mathcal{S}_t$ is then partitioned in an initial set of $K_t$ hyper-rectangles in Line 7 such that all regions contain (approximately) probability $p_{\text{thr}}$ (a parameter usually set to $1\%$ in our experiments) by iteratively cutting the axis in half and checking the inside probability until the threshold is met. This allows us to segment initial regions by probability mass (see left grid in Figure \ref{fig:grid-example}), and not by size, as is the case with uniform/equidistant grids.
We encode the resulting set of representative points corresponding to partition $\bm{\mathcal{X}}_t := \{ \mathcal{X}^{(1)}_{t}, ..., \mathcal{X}^{(K_t)}_{t} \}$ by $\bm{x}^{\text{repres}}_{t} := \{ x^{(1)}_{t}, ..., x^{(K_t)}_{t} \}$ and for each  $\mathcal{X}^{(k)}_{t}$ we select  $x^{(k)}_{t}$ as its center point.
The resulting grid is then refined in Lines 10-12 until the TV bound is smaller than a given threshold, $t\frac{\delta}{T}$, with a refinement procedure, which is detailed in subsection \ref{sec:refinement}. Note that, because at each time step the TV bound is propagated at the next time step (see Theorem \ref{thm:general-case}), this guarantees that the for each $t\in \{0,...,T\}$ we have that $\text{TV}(\mathbb{P}_{x_{t}}, \hat{\mathbb{P}}_{x_{t}})$  will be smaller than $\delta$. We should also stress that, in practice, the computational complexity of Algorithm \ref{alg:main-algo} is dominated by Step 8, which for each region requires to compute the probability mass of each component of the mixture, thus requiring a total number of integral evaluations quadratic in the size of the mixture (as in general the number of components and regions are close across consecutive iterations).

%The samples are used in Line 9 to identify a region containing at least $1-p_{thr}$ probability mass 

%For simplicity, for each time step $t$, we represent the set of representative points of the partition $\bm{\mathcal{X}}_t := \{ \mathcal{X}^{(1)}_{t}, ..., \mathcal{X}^{(K_t)}_{t} \}$ by $\bm{x}^{\text{repres}}_{t} := \{ x^{(1)}_{t}, ..., x^{(K_t)}_{t} \}$, the set of weights of a mixture is denoted by $\bm{\omega}_t := \{ \omega^{(1)}_{t}, ..., \omega^{(K_t)}_{t} \}$, and the TV upper-bound between $\mathbb{P}_{x_t}$ and $\hat{\mathbb{P}}_{x_t}$ (see Theorem \ref{thm:gaussian-case}) is denoted as $\text{TVB}(\mathbb{P}_{x_t}, \hat{\mathbb{P}}_{x_t})$. $N_s \in \mathbb{N}$ represents the number of samples to be used to build the initial grid, and $\delta > 0$ is the desired TV bound for each time step.

\RestyleAlgo{ruled}
\SetKwComment{Comment}{/* }{ */}

\begin{algorithm}[h]
\caption{Mixture approximation with TV bounds for $T < \infty$ time steps given an error threshold $\delta$}\label{alg:main-algo}
%\KwData{$n \geq 0$}
%\KwResult{$y = x^n$}
\KwIn{$\epsilon$, $p_{\text{thr}}$, $\delta$}
\KwOut{Mixture approximations $\{ \hat{\mathbb{P}}_{x_{t}} \}_{t=1}^{T}$, and TV bounds $\{ \text{TV}_{t} \}_{t=1}^{T}$}

\For{$t \in \{ 0, ..., T - 1 \}$}{
    \uIf{$t = 0$}{
        $\hat{\mathbb{P}}_{x_0} \gets \mathbb{P}_{x_0}$, $\text{TV}_0 \gets 0$
    }
    \Else{
        $\hat{\mathbb{P}}_{x_t} \gets \sum_{k = 1}^{K_{t-1}} \omega^{(k)}_{t} {\mathbb{T}}_{x^{(k)}_{t-1}}$ 
    }
    
    $\mathcal{S}_t = \texttt{IdentifyHighProbRegion}(\hat{\mathbb{P}}_{x_t}, \epsilon)$ 
    
    $\bm{\mathcal{X}}_t, \bm{x}^{\text{repres}}_{t} \gets \texttt{BuildGrid}( \mathcal{S}_t, p_{\text{thr}})$ 
    
    $\bm{\omega}_t \gets \{ \hat{\mathbb{P}}_{x_t}(\mathcal{X}^{(1)}_{t}), ..., \hat{\mathbb{P}}_{x_t}(\mathcal{X}^{(K_t)}_{t}) \}$ 
    
    $\text{TV}_{t+1} \gets \text{TV}(\mathbb{P}_{x_{t+1}}, \hat{\mathbb{P}}_{x_{t+1}}) \text{ using Thm \ref{thm:general-case}}$
    
    \While{$\text{TV}_{t+1} \geq t \frac{\delta}{T}$}{
    $\bm{\mathcal{X}}_t, \bm{x}^{\text{repres}}_{t}, \bm{\omega}_t \gets \texttt{Refine}(\bm{\mathcal{X}}_t, \bm{x}^{\text{repres}}_{t}, \bm{\omega}_t)$
    
    $\text{TV}_{t+1} \gets \text{TV}(\mathbb{P}_{x_{t+1}}, \hat{\mathbb{P}}_{x_{t+1}}) \text{ using Thm \ref{thm:general-case}}$}
}
\textbf{return } $\{ \hat{\mathbb{P}}_{x_{1}}, ..., \hat{\mathbb{P}}_{x_{T}} \}$ , $\{ \text{TV}_1, ..., \text{TV}_T \}$
\end{algorithm}

\subsection{Refinement Procedure}
\label{sec:refinement}
%In Theorem \ref{thm:convergence}, we have proved the convergence of the TV bounds to zero as the partitions become smaller. For a given threshold $\delta_t > 0$, it is possible that the initial grid built by $\texttt{BuildGrid}$ will need to be refined to attain such a bound.
The refinement procedure is described in Algorithm \ref{algo:refine-grid}. The underlying idea of the refinement procedure is to rely on Theorem \ref{thm:general-case}, where we can observe that for any $t$, $\text{TV}(\mathbb{P}_{x_{t}}, \hat{\mathbb{P}}_{x_{t}})$  depends on a sum over regions in $\bm{\mathcal{X}}_{t-1}$, where each ${\mathcal{X}^{(k)}}_{t-1} \in \bm{\mathcal{X}}_{t-1}$ contributes to the TV bound with the quantity $\max_{{x}\in \mathcal{X}^{(k)}_{t-1}} s\big( {x}, x^{(k)}_{t-1} \big)  \hat{\mathbb{P}}_{x_{t-1}}(\mathcal{X}^{(k)}_{t-1}) := c^{(k)}$. The refinement procedure boils down to partitioning regions whose contribution to the TV bound is high in relative terms. More specifically, if the contribution $c^{(k)}$ of the hypercube $\mathcal{X}^{(k)}$ is higher than a threshold $\gamma > 0$, we cut the region in $2^d$ subregions by dividing each axis by $2$ in its midpoint. An example of a refinement step is illustrated in Figure \ref{fig:grid-example}, where we can see that the refinement step decided to cut all inner regions except for the corners. 

\begin{algorithm}[h]
\SetAlgoLined
\DontPrintSemicolon
%\KwIn{$ F $\Comment*[r]{List of Sensitive Terms}}    
%\KwOut{$ S^{*} $ \Comment*[r]{Negation Excluded List}}

\SetKwFunction{FMain}{Refine}
\SetKwProg{Fn}{function}{:}{}
\Fn{\FMain{$\bm{\mathcal{X}}, \bm{x}^{\text{repres}}, \bm{\omega}$}}{
\For{$\mathcal{X}^{(k)}, x^{(k)}$ in $\bm{\mathcal{X}}, \bm{x}^{\text{repres}}$}{
    $c^{(k)} \gets \max_{x \in \mathcal{X}^{(k)}} s \big( x, x^{(k)} \big) \hat{\mathbb{P}}_{x_{t-1}} \big( \mathcal{X}^{(k)} \big)$

    \uIf{$c^{(k)} > \gamma$}{
        Cut $\mathcal{X}^{(k)}$ in $2^d$ hypercubes
    }
}
    $\bm{\omega} \gets \{ \hat{\mathbb{P}}_{x_t}(\mathcal{X}^{(1)}_{\text{new}}), ..., \hat{\mathbb{P}}_{x_t}(\mathcal{X}^{(K^{*})}_{\text{new}}) \}$ 
    
    \textbf{return} $\mathcal{X}, \bm{x}^{\text{repres}}, \bm{\omega}$ 
}
\caption{Refining the grid}
\label{algo:refine-grid}
\end{algorithm}

\begin{figure}
\centering
\includegraphics[width=0.5\textwidth]{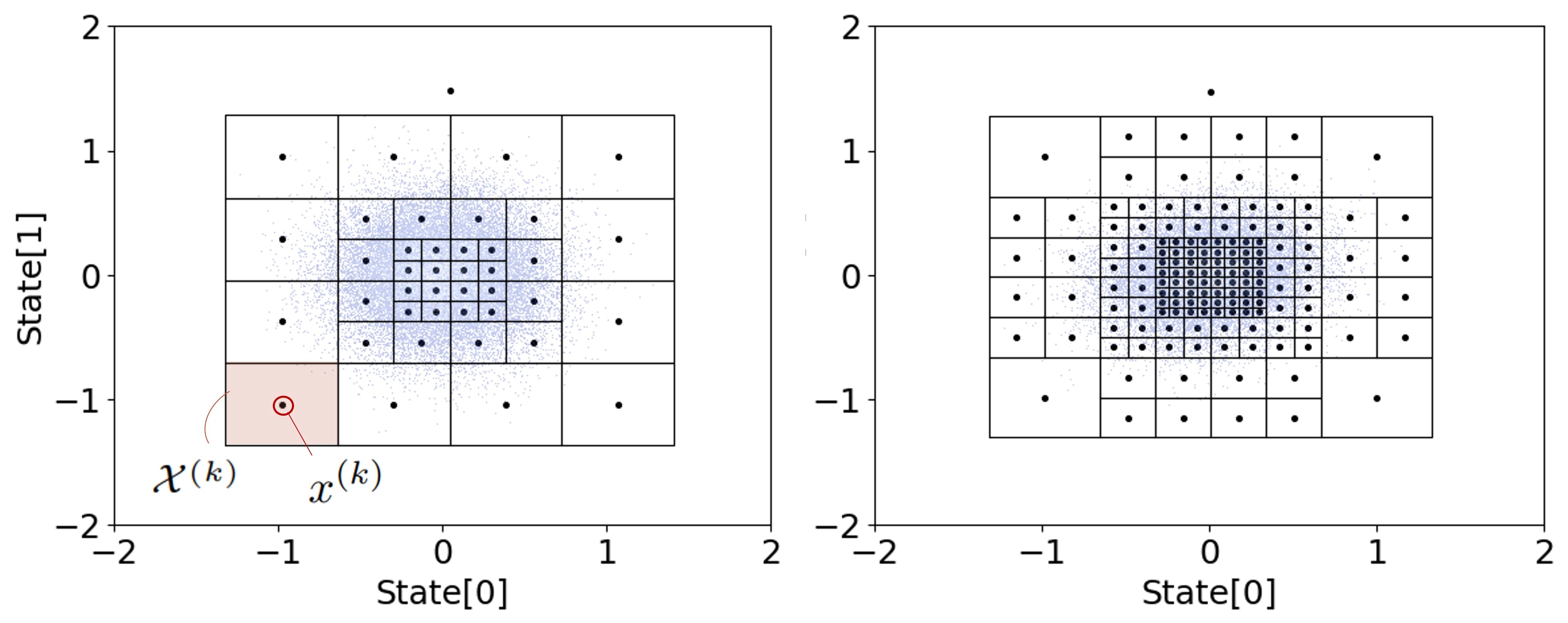}
\caption{Example of a partition used to propagate the distribution of the System presented in Example \ref{ex:gaussian-noise-explanation} with $\varepsilon \sim \mathcal{N}(0, 0.1I)$. The initial grid (left) is refined (right) with a contribution threshold of $\gamma = 1 \times 10^{-3}$. We see that the algorithm decided to cut all regions except for the corners.}
\label{fig:grid-example}
\end{figure}

\section{NUMERICAL EXPERIMENTAL RESULTS}\label{section:exp-results}
For our experiments\footnote{The code can be found at \url{https://github.com/EduardoFMDCosta/DUQviaTotalVariation}}, we consider the following benchmarks: %\AP{both as in only the last two you mention? Might be worth mentioning that the linear case is useful to analyse because there one can compute an exact solution and so allows you to evaluate the method with the ground truth?}
\begin{itemize}
\item A 2-D linear system with bimodal initial distribution with dynamics given by $g(x_t,\varepsilon_t)=Ax_t +\varepsilon_t$, where $ A = \begin{bmatrix}
    0.84 & 0.10 \\
    0.05 & 0.72
\end{bmatrix}$ and $\varepsilon_t \sim \mathcal{N}([0, 0]^T, 0.03I)$. The initial distribution is a mixture of two Gaussian distributions, that is, $x_0 \sim 0.5 \times \mathcal{N}([6, 10]^T, 0.005I) + 0.5 \times \mathcal{N}([8, 10]^T, 0.005I)$.
\item A 2-D linear system with dynamics equal to the previous one, but with noise and initial condition following a uniform distribution. Specifically, we set $\varepsilon_t \sim \mathcal{U}([-0.3, 0.3]) \times \mathcal{U}([-0.3, 0.3])$, that is, a uniform distribution for each component of the noise between $[-0.3,0.3]$. For the initial distribution we assume $x_0 \sim \mathcal{U}([-0.1, 0.1]) \times \mathcal{U}([-0.1, 0.1])$.
\item An unstable polynomial system $g(x_t,\varepsilon_t)=f(x_t) + \varepsilon_t$ where $f(x) = [f_1(x), f_2(x)]^T$ with
\begin{align*}
f_1(x) &= x_1 + 1.25 h x_2 \\
f_2(x) &= 1.4 x_2 + 0.3h (0.25x_1^2 - 0.4 x_1 x_2 + 0.25 x_2^2)
\end{align*}
for $h = 0.05$, $\varepsilon_t \sim \mathcal{N}([0, 0]^T, \sigma^2I)$, where $\sigma\in \{1,0.1, 0.001,0.0001\}$ and $x_0 \sim \mathcal{N}([1, 1]^T, 0.002I)$.
\item A 3D Dubins car model with constant velocity from \cite{mathiesen2022safety}, with $v = 5$, $h = 0.3$, and $u = 1/0.5$, with $x_0 \sim \mathcal{N}([0, 0, 0]^T, [0.005, 0.005, 0.001]])$, and $\varepsilon_t \sim \mathcal{N}([0, 0, 0]^T, [0.06, 0.06, 0.01])$.
%The initial distribution is Gaussian and given by 
%$$x_0 \sim \mathcal{N}([1, 1], 0.001I)$$.
%The noise is also Gaussian $\varepsilon_t \sim \mathcal{N}([0, 0]^T, \sigma^2_{\text{polyn}}I)$
\end{itemize}

In Algorithm \ref{alg:main-algo}, we set $p_\text{thr} = 0.01$ ($0.001$ for Dubins) to build the initial grid, and for the refinement in Algorithm \ref{algo:refine-grid} we choose $\gamma$ ranging from $10^{-6}$ to $10^{-7}$ depending on the case. All experiments were run on an Intel Core i7-1365U CPU with 16GB of RAM. Step 8
of Algorithm \ref{alg:main-algo} was parallelized on 10 cores (given that this step is the bottleneck in terms of computational complexity).

\subsection{Tightness of the Bounds}

We start our analysis from Table \ref{table:effect-variance}, where we study the efficacy of Algorithm \ref{alg:main-algo} to build a GMM approximation for the polynomial system benchmark. The table shows that unless the variance of the noise, and consequently of the one-step transition kernel, becomes very small (below $10^{-2}$) even a GMM of relatively few mixtures (order of $100$) can obtain relatively tight TV bounds
and be employed to solve Problem \ref{prob:MainProb}.
However, when the variance of the noise becomes small and the noise distribution becomes closer to a delta Dirac, as already discussed right below Corollary \ref{thm:gaussian-case}, obtaining tight TV bounds become harder. %This can be explained by looking at term $h$ in Corollary \ref{thm:gaussian-case}, which depends on the inverse of the covariance matrix. Consequently, in general, the smaller the noise variances, the more local variations of the dynamics of the system influences the TV bound. We should also stress such a behaviour, where TV bounds becomes more conservative with smaller variance, is to be expected because of the definition of TV distance, which is often not suitable for discrete support distributions \cite{}. Consequently, when the noise distribution becomes close to a Delta dirac, while 
Nevertheless, as illustrated in Figure \ref{fig:dubin-example-variance}, 
we should stress that, 
even if obtaining tight TV bounds in this setting becomes more challenging, the distribution obtained with our framework is still empirically very close to the true one even with mixtures of relatively small size. %This statement also holds for a setting with higher variance, as shown in Figure \ref{fig:dubin-example-high-variance}. 

\begin{figure}
     \centering
     \begin{subfigure}[b]{0.20\textwidth}
         \centering        \includegraphics[width=\textwidth]{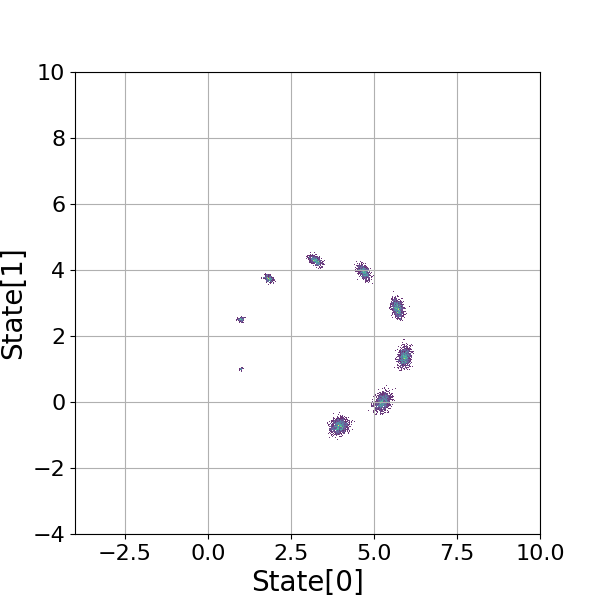}
         \caption{True distribution}
         \label{fig:dubin-low-variance-true}
     \end{subfigure}
     \hskip -2ex
     \begin{subfigure}[b]{0.20\textwidth}
         \centering       \includegraphics[width=\textwidth]{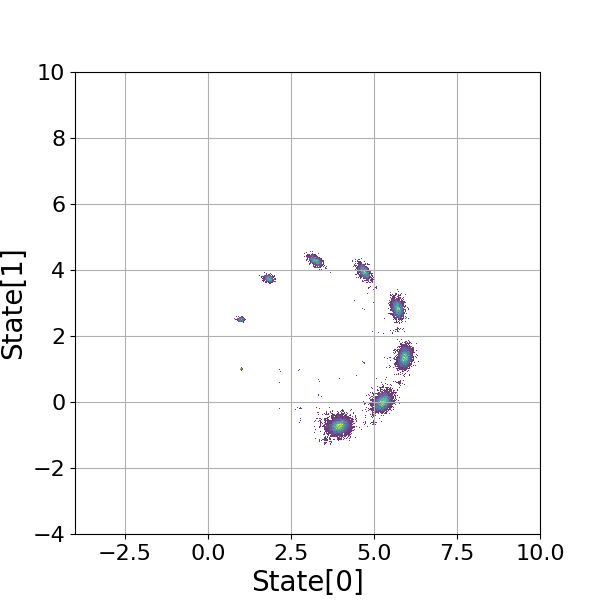}
         \caption{Mixture approximation}
         \label{fig:dubin-low-gmm-approx}
     \end{subfigure}
     \hfill
     \centering
     \begin{subfigure}[b]{0.20\textwidth}
         \centering        \includegraphics[width=\textwidth]{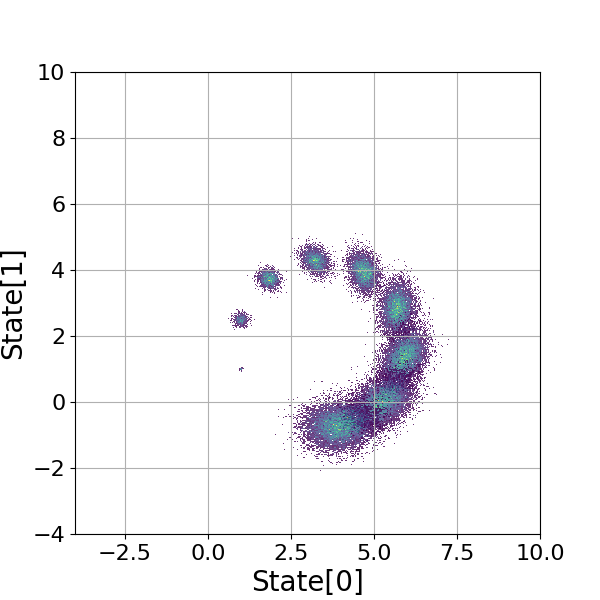}
         \caption{True distribution}
         \label{fig:dubin-high-variance-true}
     \end{subfigure}
     \hskip -2ex
     \begin{subfigure}[b]{0.20\textwidth}
         \centering       \includegraphics[width=\textwidth]{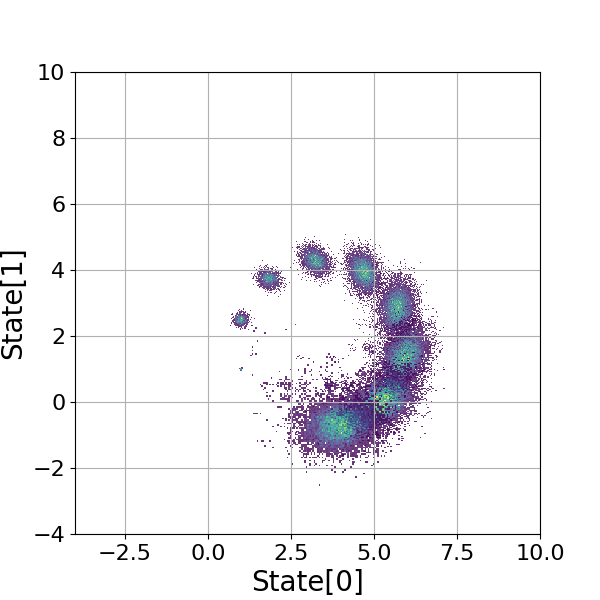}
         \caption{Mixture approximation}
         \label{fig:dubin-high-gmm-approx}
     \end{subfigure}
     \hfill
\caption{Upper row considers the Dubins car setting with low noise variance (i.e. $10^{-3}$ for the positions $x_1$ and $x_2$, and $10^{-4}$ for the steering angle $x_3$), while the lower row considers a higher variance noise ($10^{-2}$, and $2 \times 10^{-3}$, respectively). The mixtures are generated with only one refinement (see Algorithm \ref{algo:refine-grid}) with $\gamma = 10^{-6}$, resulting in mixtures of size 2500-3000. \textit{True distribution} means Monte Carlo sampling from the system, while the plots on the right are generated by sampling our mixture distributions.}
\label{fig:dubin-example-variance}
\end{figure}

In Table \ref{table:summary-bounds} we report an empirical evaluation of our approach for all four benchmarks. For all cases, we consider $5$ refinement steps and compare the mixture model obtained from Algorithm \ref{alg:main-algo} with the standard approach in \cite{landgraf2023probabilistic} (called \textit{equidistant grid}, which we will sometimes refer to the uniform grid)\footnote{Other works have proposed adaptive (i.e. non-uniform) grids for similar problems, e.g. \cite{vsimandl2006advanced}, but ours is designed with guarantees of improvement in the TV distance.}, where each distribution in the mixture is uniformly spaced in a region containing most of the probability mass of the system and obtained via samples. We consider a time horizon of $T=10$ for the bimodal case, $T=7$ for the polynomial, and $T=5$ for Dubins car and uniform settings. In all cases we can see how Algorithm \ref{algo:refine-grid} substantially outperforms the approach mentioned in \cite{landgraf2023probabilistic} by obtaining much better bounds for an equivalent number of regions. In the Dubins car, for instance, we were not able to obtain informative bounds (e.g. $\ll 1$) for a reasonable number of regions, while Algorithm \ref{algo:refine-grid} allowed for bounds between $3\%$ (first-time step) to $20\%$ (last).

\begin{table}[h]
\resizebox{\columnwidth}{!}{%
\begin{tabular}{@{}lcccccccc@{}}
\toprule
Noise variance $\sigma^2$ & \multicolumn{2}{c}{\textbf{$1$}}            & \multicolumn{2}{c}{\textbf{$10^{-1}$}}          & \multicolumn{2}{c}{\textbf{$10^{-2}$}}          & \multicolumn{2}{c}{\textbf{$10^{-3}$}} \\ \midrule
\textbf{\# refinements} & \textbf{TV} & \textbf{GMM size}           & \textbf{TV} & \textbf{GMM size}           & \textbf{TV} & \textbf{GMM size}            & \textbf{TV}   & \textbf{GMM size}  \\ \midrule
Initial grid            & 0.020       & \multicolumn{1}{c|}{121}    & 0.061       & \multicolumn{1}{c|}{124}    & 0.205       & \multicolumn{1}{c|}{115}     & 0.471         & 127                \\
1                       & 0.011       & \multicolumn{1}{c|}{485}    & 0.031       & \multicolumn{1}{c|}{497}    & 0.107       & \multicolumn{1}{c|}{461}     & 0.279         & 509                \\
2                       & 0.006       & \multicolumn{1}{c|}{1,937}  & 0.016       & \multicolumn{1}{c|}{1,982}  & 0.054       & \multicolumn{1}{c|}{1,841}   & 0.151         & 2,033              \\
3                       & 0.003       & \multicolumn{1}{c|}{7,667}  & 0.008       & \multicolumn{1}{c|}{7,868}  & 0.027       & \multicolumn{1}{c|}{7,331}   & 0.078         & 8,108              \\
4                       & 0.002       & \multicolumn{1}{c|}{26,288} & 0.005       & \multicolumn{1}{c|}{30,824} & 0.014       & \multicolumn{1}{c|}{28,997}  & 0.039         & 32,213             \\
5                       & 0.002       & \multicolumn{1}{c|}{31,070} & 0.003       & \multicolumn{1}{c|}{59,660} & 0.007       & \multicolumn{1}{c|}{112,994} & 0.020         & 127,220            \\ \bottomrule
\end{tabular}%
}
\caption{TV bounds for a 1-time step propagation of the polynomial system for the same initial distribution $x_0$ for various values of variance $\sigma^2$. In each row, we show an iteration of the refinement algorithm and report both the resulting TV bound and the size of the approximating GMM. }
\label{table:effect-variance}
\end{table}

\begin{table}[h]
\resizebox{\columnwidth}{!}{%
\begin{tabular}{lll|ll|ll|ll}
\textbf{}                  & \multicolumn{2}{c|}{\textbf{Bimodal}}                                      & \multicolumn{2}{c|}{\textbf{Polynomial}}                                   & \multicolumn{2}{c|}{\textbf{Uniform}}                                      & \multicolumn{2}{c}{\textbf{Dubins}}                                       \\ \hline
\textbf{Grid method} & \multicolumn{1}{c}{\textbf{Alg 2}} & \multicolumn{1}{c|}{\textbf{Equid}} & \multicolumn{1}{c}{\textbf{Alg 2}} & \multicolumn{1}{c|}{\textbf{Equid}} & \multicolumn{1}{c}{\textbf{Alg 2}} & \multicolumn{1}{c|}{\textbf{Equid}} & \multicolumn{1}{c}{\textbf{Alg 2}} & \multicolumn{1}{c}{\textbf{Equid}} \\ \hline
$\text{TV}_1$                     & 0.004                               & 0.007                                & 0.004                               & 0.006                                & 0.004                               & 0.003                                & 0.028                               & 0.99                                \\
$\text{TV}_T$                     & 0.061                               & 0.092                                & 0.099                               & 0.179                                & 0.041                               & 0.048                                & 0.198                               & 1.00                                \\ \hline
Avg TV                     & 0.033                               & 0.049                                & 0.039                               & 0.070                                & 0.022                               & 0.022                                & 0.101                               & 1.00                                \\ \hline
\end{tabular}%
}
\caption{TV bounds for the first and last propagation steps (as well as for the average) for each system. The bounds are computed from grids generated by Algorithm \ref{algo:refine-grid} (Alg \ref{algo:refine-grid}) with 5 refinements and by the uniform/equidistant grids (Equid) of equivalent size.}
\label{table:summary-bounds}
\end{table}

\subsection{Application to chance constrained certification}

To highlight the usefulness of our approach, we show how it can be employed to quantify the probability that, at each time step, System \eqref{eq:original-problem} satisfies probabilistic constraints, a setting commonly occurring in planning problems for stochastic systems \cite{adams2022formal,brudigam2021stochastic,ho2022gaussian}. In particular, for the bimodal system, we consider the setting in Figure \ref{fig:bimodal-chance-constraint}, where the goal is to quantify the probability that at each time step the system will enter the unsafe region in red (which we denote by $p_{\text{hit}}$. In Table \ref{table:chance-constraint}, we report both the empirical value obtained by simulating the specific system $10^4$ times and the upper bound of the hitting probability at each time step according to our framework (computed as the probability that the GMM approximation hits the obstacle + the TV bound at the particular time step computed according to Theorem \ref{thm:general-case}). We observe that the upper bound for $p_{\text{hit}}$ is fairly close to the actual hitting probability of the system, due to two factors: i) the hitting probability estimated by our mixture is very close to the true value (analogous to what is shown e.g. in Figure \ref{fig:dubin-example-variance}), and ii) the TV bounds are reasonably close to zero (see also Table \ref{table:summary-bounds}).

\begin{table}[h]
\resizebox{\columnwidth}{!}{%
\begin{tabular}{lcc}
& \multicolumn{2}{c}{\textbf{Bimodal}} \\ \hline
\textbf{Time step} & \textbf{Actual hitting probabity $p_{\text{hit}}$}  & \textbf{Upper bound $p_{\text{hit}}$}    \\ \hline
1                  & 0.0               & 0.4              \\
2                  & 0.0               & 0.9              \\
3                  & 0.0               & 1.6              \\
4                  & 0.0               & 2.4              \\
5                  & 18.4              & 21.5             \\
6                  & 38.0              & 42.1             \\
7                  & 20.9              & 25.6             \\
8                  & 3.5               & 8.5              \\
9                  & 0.0               & 5.8              \\
10                 & 0.0               & 6.2              \\ \hline
\end{tabular}%
}
\caption{Actual $p_{\text{hit}}$ is the unsafe set hitting probability computed via Monte Carlo simulation of the actual system ($10^4$ samples), while the upper bound is given by the probability mass of the GMM inside the unsafe set + TV bound.}
\label{table:chance-constraint}
\end{table}

\begin{figure}
    \centering
     \begin{subfigure}[b]{0.24\textwidth}
         \centering        \includegraphics[width=\textwidth]{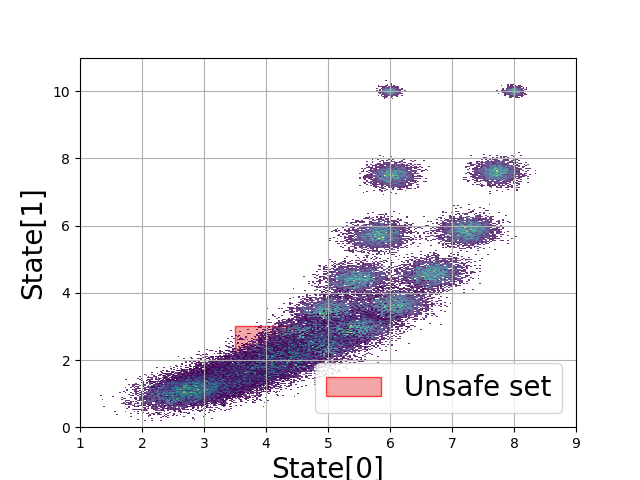}
         \caption{True distribution}
         \label{fig:bimodal-true}
     \end{subfigure}
     \hskip -2ex
     \begin{subfigure}[b]{0.24\textwidth}
         \centering       \includegraphics[width=\textwidth]{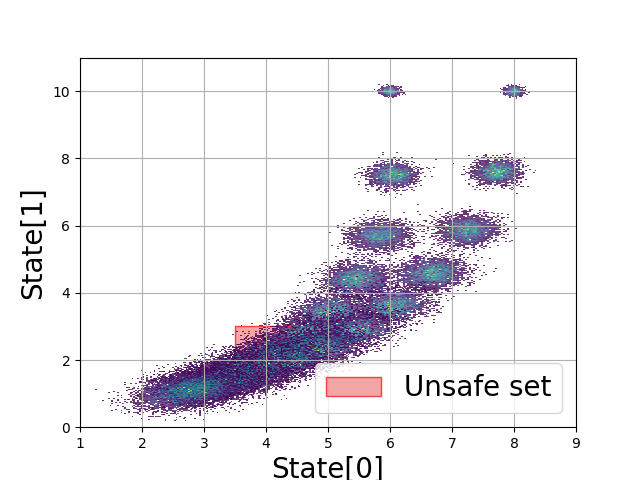}
         \caption{Mixture approximation}
         \label{fig:bimodal-gmm-approx}
     \end{subfigure}
     \hfill
\caption{Bimodal setting for which we present the hitting probability study. The left image shows a Monte Carlo simulation of the system, while the right displays samples from our mixture approximation. $p_{\text{hit}}$ is computed w.r.t. the red set.}
\label{fig:bimodal-chance-constraint}
\end{figure}

\section{PROOFS} \label{section:proofs}

\subsection{Proof of Theorem \ref{thm:general-case}}

For the proof we introduce a new auxiliary distribution $\Tilde{\mathbb{P}}_{x_t}$, whose density is defined as
$$\Tilde{p}_{x_t} (y) := \int_{\mathbb{R}^d} p(y \mid x) \hat{p}_{x_{t-1}}(x) dx.$$
That is, $\Tilde{\mathbb{P}}_{x_t}$ is the probability distributions given by the propagation of mixture $\hat{P}_{x_{t-1}}$ through the dynamics of the system. Then, the proof uses the fact that by the triangle inequality
\begin{align*}\text{TV} &\left( \mathbb{P}_{x_t}, \hat{\mathbb{P}}_{x_t} \right) \leq \text{TV} \left( \mathbb{P}_{x_t}, \Tilde{\mathbb{P}}_{x_t} \right) + \text{TV} \left( \Tilde{\mathbb{P}}_{x_t}, \hat{\mathbb{P}}_{x_t} \right).
\end{align*}
Consequently, to bound $\text{TV} \left( \mathbb{P}_{x_t}, \hat{\mathbb{P}}_{x_t} \right) $ it is enough to bound $\text{TV} \left( \mathbb{P}_{x_t}, \Tilde{\mathbb{P}}_{x_t} \right) $ and $\text{TV} \left( \Tilde{\mathbb{P}}_{x_t}, \hat{\mathbb{P}}_{x_t} \right).$
We start from $\text{TV} \left( \Tilde{\mathbb{P}}_{x_t}, \hat{\mathbb{P}}_{x_t} \right)$. To do that we first note that for every $y \in \mathbb{R}^d$, the following holds:
\begin{align*}
|\Tilde{p}_{x_t}(y) &- \hat{p}_{x_t}(y)|\\ & = \Bigl| \int_{\mathbb{R}^d} p(y \mid x) \hat{p}_{x_{t-1}}(x) dx - \sum_{k=1}^{K_{t-1}} \omega^{(k)}_{t} p(y \mid x^{(k)}_{t-1})  \Bigr|  \\
& = \Bigl| \sum_{k=1}^{K_{t-1}} \int_{\mathcal{X}^{(k)}_{t-1}} \big( p(y \mid x) - p(y \mid x^{(k)}_{t-1}) \big) \hat{p}_{x_{t-1}}(x) dx \Bigr| \\
&\leq \sum_{k=1}^{K_{t-1}} \int_{\mathcal{X}^{(k)}_{t-1}} \big| p(y \mid x) - p(y \mid x^{(k)}_{t-1}) \big| \hat{p}_{x_{t-1}}(x) dx,
\end{align*}
where the second line uses the definition of $\omega^{(k)}_{t}$ in Eqn \eqref{eq:approximation-scheme}, and the third the triangle inequality.
From this, it follows that 
\begin{equation*}
\begin{split}
& \text{TV} \left( \Tilde{\mathbb{P}}_{x_t}, \hat{\mathbb{P}}_{x_t} \right) = \frac{1}{2} \int_{\mathbb{R}^d} \big| \Tilde{p}_{x_t}(y) - \hat{p}_{x_t}(y) \big| dy \\
&\leq \frac{1}{2} \int_{\mathbb{R}^d} \left[ \sum_{k=1}^{K_{t-1}} \int_{\mathcal{X}^{(k)}_{t-1}} \big| p(y \mid x) - p(y \mid x^{(k)}_{t-1}) \big| \hat{p}_{x_{t-1}}(x) dx \right] dy \\
&= \sum_{k=1}^{K_{t-1}} \int_{\mathcal{X}^{(k)}_{t-1}} \underbrace{\left[ \frac{1}{2} \int_{\mathbb{R}^d} |p(y \mid x) - p(y \mid x^{(k)}_{t-1})| dy \right]}_{s \big( x, x^{(k)}_{t-1} \big)} \hat{p}_{x_{t-1}}(x) dx \\
&\leq \sum_{k=1}^{K_{t-1}} \max_{x \in \mathcal{X}^{(k)}_{t-1}} s \big( x, x^{(k)}_{t-1} \big) \int_{\mathcal{X}^{(k)}_{t-1}} \hat{p}_{x_{t-1}}(x) dx \\
&= \sum_{k=1}^{K_{t-1}} \max_{x \in \mathcal{X}^{(k)}_{t-1}} s \big( x, x^{(k)}_{t-1} \big) \hat{\mathbb{P}}_{x_{t-1}}\big( \mathcal{X}^{(k)}_{t-1} \big)
\label{eq:bound-tv-before-gaussian-assumption}
\end{split}
\end{equation*}
where the change order of integration in the third line is possible due to the application of Fubini's theorem after the observation that the inner function is integrable on $\mathbb{R}^d \times \mathbb{R}^d$.

Now, to conclude, for the TV distance between $\mathbb{P}_{x_t}$ and $\Tilde{\mathbb{P}}_{x_t}$ we have that
\begin{equation*}
\begin{split}
\text{TV} & \left( \mathbb{P}_{x_t}, \Tilde{\mathbb{P}}_{x_t} \right) \\
& = \frac{1}{2} \int_{\mathbb{R}^d} \big| p_{x_t}(y) - \Tilde{p}_{x_t}(y) \big| dy \\
&= \frac{1}{2} \int_{\mathbb{R}^d} \big| \int_{\mathbb{R}^d} p(y \mid x) p_{x_{t-1}}(x) - \int_{\mathbb{R}^d} p(y \mid x) \hat{p}_{x_{t-1}}(x) dx \big| dy \\
&= \frac{1}{2} \int_{\mathbb{R}^d} \big| \int_{\mathbb{R}^d} p(y \mid x)(p_{x_{t-1}}(x) - \hat{p}_{x_{t-1}}(x))dx \big| dy \\
&\leq \frac{1}{2} \int_{\mathbb{R}^d} \int_{\mathbb{R}^d} p(y \mid x) \big| p_{x_{t-1}}(x) - \hat{p}_{x_{t-1}}(x) \big| dx dy \\
&= \frac{1}{2} \int_{\mathbb{R}^d} \underbrace{\left[ \int_{\mathbb{R}^d} p(y \mid x) dy \right]}_{= 1} \big| p_{x_{t-1}}(x) - \hat{p}_{x_{t-1}}(x) \big| dx \\
&= \frac{1}{2} \int_{\mathbb{R}^d} \big| p_{x_{t-1}}(x) - \hat{p}_{x_{t-1}}(x) \big| dx \\
&= \text{TV}\left( \mathbb{P}_{x_{t-1}}, \hat{\mathbb{P}}_{x_{t-1}} \right)
\end{split}
\end{equation*}

where the fourth line uses the triangle inequality, and the fifth consists of a change in the integration order applying Fubini's theorem.

\subsection{Proof of Corollary \ref{thm:gaussian-case}}

From the definition of total variation, we have that
$$
s \big( x, x^{(k)}_{t-1} \big) = \text{TV} \big( \mathcal{N}(f(x), \Sigma_{\varepsilon}), \mathcal{N}(f(x^{(k)}_{t-1}), \Sigma_{\varepsilon}) \big).
$$
It then follows from \cite{devroye2018total} (or Theorem 1 in \cite{barsov1987estimates}), that  
\begin{equation*}
\begin{split}
s& \big( x, x^{(k)}_{t-1} \big) \\ 
& = \mathbb{P} \left( \mathcal{N}(0, 1) \in \left[ - \sqrt{2} h \big( x, x^{(k)}_{t-1} \big), \sqrt{2} h \big( x, x^{(k)}_{t-1} \big) \right] \right) \\
& = \Phi \left( \sqrt{2} h \big( x, x^{(k)}_{t-1} \big) \right) - \Phi \left( -\sqrt{2} h \big( x, x^{(k)}_{t-1} \big) \right) \\
& = 2 \Phi \left( \sqrt{2} h \big( x, x^{(k)}_{t-1} \big) \right) - 1  = \text{erf} \left( h \big( x, x^{(k)}_{t-1} \big) \right)
\end{split}
\end{equation*}

\subsection{Proof of Theorem \ref{thm:convergence}}

\begin{proof}
From Theorem \ref{thm:general-case} we have that $$\text{TV} \left( \mathbb{P}_{s}, \hat{\mathbb{P}}_{s} \right) \leq \sum_{t=1}^{s} \sum_{k=1}^{K_{t-1}} \max_{x \in \mathcal{X}^{(k)}_{t-1}} s \big( x, x^{(k)}_{t-1} \big) \hat{\mathbb{P}}_{x_{t-1}} \big( \mathcal{X}^{(k)}_{t-1} \big). $$

Consequently, it is enough to prove that for any $\bar\epsilon >0$, for any $t$, there exists a partition of $\mathbb{R}^d =\{\mathcal{X}^{(1)},...,\mathcal{X}^{(K+1)} \}$ in $K + 1$ sets such that 
\begin{equation}
\sum_{k=1}^{K + 1} \max_{x \in \mathcal{X}^{(k)}} s \big( x, x^{(k)} \big) \hat{\mathbb{P}}_{x_{t-1}} \big( \mathcal{X}^{(k)} \big) \leq \bar\epsilon.
\end{equation}
In fact, we can then conclude by taking $\bar\epsilon < \frac{\epsilon}{t}.$
In order to do that, we first consider a hypercube $\mathcal{S}$ of diameter $r$ w.r.t. the infinity norm containing at least $1-\frac{\bar\epsilon}{2}$ probability mass of $\hat{\mathbb{P}}_{x_{t-1}}$, which is guaranteed to exist as $\hat{\mathbb{P}}_{x_{t-1}}$ is almost surely bounded by assumption.  We then select $\mathcal{X}^{(K+1)} = \mathbb{R}^d \setminus \mathcal{S}$ and let $\{\mathcal{X}^{(1)},...,\mathcal{X}^{(K)} \}$ be a uniform partition of $\mathcal{S}$ in $K$ hypercubic regions each of diameter $\frac{r}{\sqrt[d]{K}}$. What is left to show is that there exists $\bar{K}$ such that for any $K\geq \bar{K}$

\begin{align}
\sum_{k=1}^{K}  \max_{x \in \mathcal{X}^{(k)}} s \big( x, x^{(k)} \big) \hat{\mathbb{P}}_{x_{t-1}} \big( \mathcal{X}^{(k)} \big) \leq \frac{\bar \epsilon}{2}.
\end{align}

Now, note that:

\begin{align*}
 \sum_{k=1}^{K} &  \max_{x \in \mathcal{X}^{(k)}} s \big( x, x^{(k)} \big)  \hat{\mathbb{P}}_{x_{t-1}} \big( \mathcal{X}^{(k)} \big) \\  
& \leq \max_{ x \in \mathcal{X}^{(k)}} \max_{\bar{x}: \|x-\bar{x}\|_{\infty} \leq \frac{r}{\sqrt[d]{K}}} s \big( x, \bar{x} \big)  \sum_{k=1}^K \hat{\mathbb{P}}_{x_{t-1}} \big( \mathcal{X}^{(k)} \big) \\
& \leq \max_{ x \in \mathcal{X}^{(k)}} \max_{\bar{x}: \|x-\bar{x}\|_{\infty} = \frac{r}{\sqrt[d]{K}}} s \big( x, \bar{x} \big) \cdot 1.
\end{align*}

To conclude it is enough to note that, for any $x \in \mathcal{X}^{(k)}$, if $p(y \mid \bar{x}) \rightarrow p(y \mid x)$ when $K \rightarrow \infty$ for every $y$ a.s., then it holds that
$$ \lim_{K\to \infty} \left( \max_{\bar{x}: \|x-\bar{x}\|_{\infty} \leq \frac{r}{\sqrt[d]{K}}} s \big( x, \bar{x} \big)  \right) = 0, $$
since $\lim_{K\to \infty} s \big( x, \bar{x} \big)  = 0$ due to Scheffe's Lemma \cite{driver2007math} (Lemma 13.3).
\end{proof}

\section{CONCLUSION}
We introduced a framework to approximate the distribution of a non-linear stochastic dynamical system over time with a mixture of distributions. By deriving bounds for the Total Variation distance, we showed that for each time step, the resulting approximating mixture distribution is guaranteed to be $\delta-$close to the distribution of the system. On a set of experiments, we showed the efficacy of our framework. In our view, there are several opportunities for improvement of the current framework, such as i) mixture compression (reducing the number of components in the approximating mixture with formal guarantees) to decrease the computational complexity, ii) bound tightening (in Theorem \ref{thm:general-case}, we are working on ways to improve the bound without resorting to the maximization), and iii) extension of the formal guarantees to other metrics on probability distributions, such as the Wasserstein distance.

%%%%%%%%%%%%%%%%%%%%%%%%%%%%%%%%%%%%%%%%%%%%%%%%%%%%%%%%%%%%%%%%%%%%%%%%%%%%%%%%
\section{ACKNOWLEDGEMENTS}

L.L. and E.F. are partially supported by the NWO (grant OCENW.M.22.056).

%%%%%%%%%%%%%%%%%%%%%%%%%%%%%%%%%%%%%%%%%%%%%%%%%%%%%%%%%%%%%%%%%%%%%%%%%%%%%%%%

\bibliographystyle{IEEEtran}
%\bibliography{biblio}
\input{main.bbl}

\addtolength{\textheight}{-3cm}   % This command serves to balance the column lengths

%%%%%%%%%%%%%%%%%%%%%%%%%%%%%%%%%%%%%%%%%%%%%%%%%%%%%%%%%%%%%%%%%%%%%%%%%%%%%%%%

\end{document}